\documentclass[submission,copyright,creativecommons]{eptcs}
\providecommand{\event}{QPL 2023} 
\usepackage{breakurl}             
\usepackage{underscore}           

\usepackage{scalerel}

\DeclareMathAlphabet{\mathcal}{OMS}{cmsy}{m}{n}
\usepackage{amssymb}
\usepackage{amsthm} 
\usepackage{amsmath}
\usepackage{amsfonts}
\usepackage{tikz-cd}
\usepackage{mathrsfs}
\usepackage{mathtools}
\usepackage{array}
\usepackage{extarrows}
\usepackage{bm}
\usepackage[b]{esvect}

\tikzset{%
    symbol/.style={%
        draw=none,
        every to/.append style={%
            edge node={node [sloped, allow upside down, auto=false]{$#1$}}}
    }
}

\newlength\bshft
\def\fakebold#1{\ThisStyle{\ooalign{$\SavedStyle#1$\cr%
  \kern-\bshft$\SavedStyle#1$\cr%
  \kern\bshft$\SavedStyle#1$}}}

\def\mybolddd#1{\fakebold{\fakebold{\fakebold{#1}}}}
\def\myboldd#1{\fakebold{\fakebold{#1}}}

\newtheorem{theorem}{Theorem}[section]
\newtheorem{lemma}[theorem]{Lemma}
\newtheorem{proposition}[theorem]{Proposition}
\newtheorem{corollary}[theorem]{Corollary}
\newtheorem{conjecture}[theorem]{Conjecture}

\newtheorem{definition}[theorem]{Definition}

\theoremstyle{remark}

\newtheorem{example}[theorem]{Example}

\newcommand{\A}{\mathcal A}

\newcommand{\D}{\mathcal D }

\renewcommand{\H}{\mathcal H}

\renewcommand{\P}{\mathcal P}
\newcommand{\Q}{\mathcal Q}

\newcommand{\U}{\mathcal U}

\newcommand{\W}{\mathcal W}
\newcommand{\X}{\mathcal X}
\newcommand{\Y}{\mathcal Y}
\newcommand{\Z}{\mathcal Z}

\newcommand{\CC}{\mathbb C}

\newcommand{\kK}{\mathfrak K}

\newcommand{\op}{\mathrm{op}}
\newcommand{\ran}{\mathop{\mathrm{ran}}}

\newcommand{\down}{\downarrow}
\newcommand{\qo}{{\mybolddd{\preccurlyeq}}}
\newcommand{\qop}{{\mybolddd{\succcurlyeq}}}
\newcommand{\qsubseteq}{{\mybolddd{\subseteq}}}
\newcommand{\qsupseteq}{{\mybolddd{\supseteq}}}
\newcommand{\qun}{{\mybolddd{\bigcup}}}
\newcommand{\qni}{{\mybolddd{\ni}}}
\newcommand{\qdown}{{\myboldd{{\down}\{\cdot\}}}}

\newcommand{\qnotni}{{\mybolddd{\not\ni}}}
\newcommand{\qsup}{{\mybolddd{\bigvee}}}

\newcommand{\atomof}{\text{\raisebox{1.0pt}{\,$\propto$\,}}}

\newcommand{\At}{\mathrm{At}}

\newcommand{\Eval}{\mathrm{Eval}}

\newcommand{\qPow}{\mathrm{qPow}}
\newcommand{\Pow}{\mathrm{Pow}}
\newcommand{\qDwn}{\mathrm{qDwn}}
\newcommand{\qUp}{\mathrm{qUp}}
\newcommand{\Dwn}{\mathrm{Dwn}}

\newcommand{\Set}{\mathbf{Set}}
\newcommand{\POS}{\mathbf{Pos}}
\newcommand{\CPO}{\mathbf{CPO}}
\newcommand{\qRel}{\mathbf{qRel}}
\newcommand{\qSet}{\mathbf{qSet}}
\newcommand{\qPOS}{\mathbf{qPos}}
\newcommand{\qPos}{\mathbf{qPos}}
\newcommand{\qSup}{\mathbf{qSup}}
\newcommand{\qCPO}{\mathbf{qCPO}}
\newcommand{\Sup}{\mathbf{Sup}}

\newcommand{\Rel}{\mathbf{Rel}}

\newcommand{\RelPos}{\mathbf{RelPos}}
\newcommand{\qRelPos}{\mathbf{qRelPos}}

\DeclareMathAlphabet{\mathpzc}{OT1}{pzc}{m}{it}

\title{Quantum Suplattices}
\author{Gejza Jen\v{c}a
\institute{Slovak University of Technology,\\
Bratislava, Slovak Republic}
\email{gejza.jenca@stuba.sk}
\and
Bert Lindenhovius 
\institute{Slovak Academy of Sciences,\\
Bratislava, Slovak Republic}
\email{lindenhovius@mat.savba.sk}
}

\def\titlerunning{Quantum Suplattices}
\def\authorrunning{G. Jen\v{c}a \& B. Lindenhovius}

\allowdisplaybreaks

\begin{document}

\maketitle

\begin{abstract}
Building on the theory of quantum posets, we introduce a non-commutative version of suplattices, i.e., complete lattices whose morphisms are supremum-preserving maps, which form a step towards a new notion of quantum topological spaces. We show that the theory of these `quantum suplattices' resembles the classical theory: the opposite quantum poset of a quantum suplattice is again a quantum suplattice, and quantum suplattices arise as algebras of a non-commutative version of the monad of downward-closed subsets of a poset. The existence of this monad is proved by introducing a non-commutative generalization of monotone relations between quantum posets, which form a compact closed category. Moreover, we introduce a non-commutative generalization of Galois connections and we prove that an upper Galois adjoint of a monotone map between quantum suplattices exists if and only if the map is a morphism of quantum suplattices. Finally, we prove a quantum version of the Knaster-Tarski fixpoint theorem: the quantum set of fixpoints of a monotone endomap on a quantum suplattice form a quantum suplattice. 
\end{abstract}

\section{Introduction}\label{sec:intro}

A poset is called a \emph{complete lattice} if it has all suprema, or equivalently, if it has all infima. However, a function between complete lattices that preserves all suprema does not necessarily preserve all infima. Hence, one can define several categories of complete lattices with different classes of morphisms. For instance, the class consisting of maps that preserve both all suprema and all infima, or the class of maps that only preserve all suprema. If we choose this latter class of morphisms, we typically call the objects of the resulting category $\Sup$ \emph{suplattices} instead of complete lattices. 

In this contribution, we introduce a noncommutative version of suplattices, which we call \emph{quantum suplattices}. One of the reasons why we are interested in these objects is that they might lead to a notion of quantum topological spaces that allows for the quantization of topological spaces that are not locally compact Hausdorff, such as the Scott topology on a dcpo. Any topology of a topological space is in particular a complete lattice; the usual approach to quantum topological spaces are C*-algebras, which can be regarded as noncommutative locally compact Hausdorff spaces.

The approach we take is the program of \emph{discrete quantization} \cite{Kornell20}. Here, \emph{quantizing} some mathematical structure is understood as the operation of finding a noncommutative generalization or version of the structure. This can be done by internalizing the structure in a suitable category of operator algebras. For discrete quantization, this category is called $\qRel$, which is equivalent to the category of von Neumann algebras isomorphic to some (possibly infinite) $\ell^\infty$-sum of matrix algebras (such von Neumann algebras are also called \emph{hereditarily atomic}) equipped with Weaver's quantum relations \cite{Weaver10}. The category $\qRel$ shares many properties with $\Rel$. For instance, it is also dagger compact \cite{AbramskyCoecke08,coeckekissinger,heunenvicary} and enriched over $\Sup$. Since many mathematical structures can be described in terms of the dagger structure and the $\Sup$-enrichment of $\Rel$, this makes $\qRel$ a very suitable tool for quantization. Since hereditarily atomic von Neumann algebras have a discrete character, they can be regarded as noncommutative sets, which explains the name `discrete quantization'. 

Regarding discrete quantization, one could argue that it is a disadvantage that we do not work in full generality with all von Neumann algebras. However, we see this lack of generality as a feature, not as a bug: the category of all von Neumann algebras and quantum relations is not compact, whereas $\qRel$ is. This is of huge importance for the theory of quantum suplattices, which relies heavily on $\qRel$ being compact. 
By definition, any matrix algebra $\mathrm{M}_d(\mathbb C)$, which is often used to represent a qudit, is an example of a hereditarily atomic von Neumann algebra. Since any tensor product of two matrix algebras is a hereditarily atomic von Neumann algebra, systems of multiple qudits can be represented by hereditarily atomic von Neumann algebras. Therefore, hereditarily atomic von Neumann algebras are sufficient for most practical applications in quantum information theory and in quantum computing.  Recently, discrete quantization was applied successfully in the denotational semantics of quantum programming languages \cite{KLM20}.

\subsection{Related work}
Discrete quantization can be regarded as a special case of quantization via quantum relations, which we distilled by Weaver \cite{Weaver10} from his work with Kuperberg on quantum metric spaces \cite{kuperbergweaver:quantummetrics}. The category $\qRel$ whose objects are called \emph{quantum sets} was introduced by Kornell \cite{Kornell18}, who showed that this category is dagger compact. Moreover, in the same reference, he quantized  functions by internalizing them in $\qRel$, and showed showed that resulting category $\qSet$ of quantum sets and quantized functions is symmetric monoidal closed, complete, cocomplete, and dual to the category of hereditarily atomic von Neumann algebras and normal $*$-homomorphisms. Furthermore, Kornell showed in \cite{Kornell20} that $\qRel$ is equivalent to the category of hereditarily atomic von Neumann algebras and Weaver's quantum relations, and introduced a logic with equality for $\qRel$. 

Quantum posets were already defined by Weaver in \cite{Weaver10}. The properties of the category of quantum posets in the framework of discrete quantum mathematics were investigated in \cite{KLM22} by Kornell, Mislove and the second author. The same authors proceeded in \cite{KLM20} by quantizing cpos by means of discrete quantization. Traditionally, cpos are a class of posets that form the essential objects for the denotational semantics of programming languages. One would expect that the denotational semantics of quantum programming languages will require quantized cpos. The current state-of-the art quantum programming language is Proto-Quipper-M, which was introduced by Rios and Selinger \cite{pqm-small}, subsequently extended with recursive terms in \cite{LMZ18} and then with recursive types in \cite{LMZ20} by Mislove, Zamdzhiev and the second author. In \cite{KLM20}, the quantized cpos were successfully used to construct sound and adequate denotational models for these extensions.

Finally, we mention the quantum graphs \cite[Definition 2.6(d)]{Weaver10}\cite{DuanSeveriniWinter}, which recently  attracted some attention \cite{ATSERIASetal,Stahlke, Weaver3, MustoReutterVerdon, BrannanGanesanHarris, ChirvasituWasilewski}. Quantum graphs can be described in the framework of discrete quantum quantization in a similar way as quantum posets. Just like a poset is a special kind of graph, a quantum poset is a special kind of quantum graph, and many concepts and techniques from quantum graphs carry over to quantum posets.

\subsection{Overview of the paper}
We start by giving a recap of quantum sets and quantum posets. Then we introduce monotone relations between quantum sets. These are of importance, because the ordinary down-set monad on the category $\POS$ of ordinary posets can be obtained from an adjunction between $\POS$ and the category $\RelPos$ of  posets and monotone relations. 
We proceed by introducing the quantum down-set monad $\qDwn$ and explaining its connection with upper sets and the quantum power set. We then introduce a quantum generalization of Galois connections, which we use to define quantum suplattices. We show that $\qDwn(\X)$ is a quantum suplattice for any quantum poset $\X$. Together with the characterization of Galois connections between quantum posets (cf. Theorem \ref{thm:Galois}), these are the only proofs we include, just to give a flavor of how to work with quantum sets. Furthermore, we sketch why the opposite of a quantum suplattice is also a quantum suplattice. Finally, we discuss enrichment over $\Sup$ and fixpoints of monotone endomaps between quantum suplattices.

Let us stress that although the quantized theorems we included in the present paper are almost verbatim copies of their classical versions, the proof of a quantized theorem is usually much more complicated than its classical counterpart.

\section{Preliminaries}
\subsection{Quantum sets}
The basic reference for this section is \cite{Kornell18}. 
Here, the definition of $\qRel$ implicitly makes use of categorical constructions, which we choose to highlight. In order to do so, we first introduce the category $\mathbf{FdOS}$ whose objects are nonzero finite-dimensional Hilbert spaces. A morphism  $A:X\to Y$ in $\mathbf{FdOS}$ is a \emph{operator space}, that is a subspace of the vector space $L(X,Y)$ of linear operators $X\to Y$. We define composition of $A$ with a morphism $B:Y\to Z$ in $\mathbf{FdOS}$ as the operator space $B\cdot A:=\mathrm{span}\{ba:a\in A,b\in B\}$. The identity morphism on $X$ is the operator space $\CC 1_X:=\{\lambda 1_X:\lambda\in\CC\}$.  Since the space $L(X,Y)$ is actually a finite-dimensional Hilbert space itself via the inner product $(a,b)\mapsto\mathrm{Tr}(a^\dag b)$, where $a^\dag:Y\to X$ denotes the hermitian adjoint of $a\in L(X,Y)$, the homset $\mathbf{FdOS}(X,Y)$ becomes a complete modular ortholattice; the order on the homset is explicitly given by $A\leq B$ if and only if $A$ is a subspace of $B$. Since composition in $\mathbf{FdOS}(X,Y)$ preserves suprema, $\mathbf{FdOS}$ is enriched over the category $\mathbf{Sup}$ of complete lattices and suprema-preserving functions; any such category is also called a  \emph{quantaloid}.

Products and coproducts in quantaloids coincide and are also called \emph{sums}. Any quantaloid $\mathbf Q$ has a  free sum-completion of $\mathbf Q$, which can be described by the quantaloid $\mathrm{Matr}(\mathbf Q)$, whose objects are $\mathbf{Set}$-indexed families of objects $(X_i)_{i\in I}$ of $\mathbf Q$, and whose morphisms $R:(X_i)_{i\in I}\to(Y_j)_{j\in J}$ are `matrices' whose $(i,j)$-component $R(i,j)$ is a $\mathbf Q$-morphism $X_i\to Y_j$. Composition in $\mathrm{Mat}(\mathbf Q)$ is defined by matrix multiplication: we define $S\circ R$ for $S:(Y_j)_{j\in J}\to (Z_k)_{k\in K}$ by $(S\circ R)(i,k)=\bigvee_{j\in J}S(j,k)\cdot R(i,j)$ for each $i\in I$ and $k\in K$, where $\cdot$ denotes the composition of morphisms in $\mathbf Q$. The $(i,i')$-component of the identity morphism on an object $(X_i)_{i\in I}$ is the identity morphism of $X_i$ if $i=i'$, and $0$ otherwise. The order on the homsets of $\mathrm{Matr}(\mathbf Q)$ is defined componentwise. The matrix-construction as the free sum-completion of quantaloids was introduced in \cite{Heymans-Stubbe}, and is a special case of a matrix-construction for more general bicategories as described in \cite{BCSW}.

We now define $\mathbf{qRel}$ as the quantaloid $\mathrm{Matr}(\mathbf{FdOS})$. Any object $\X$ of $\qRel$ is called a \emph{quantum set}; whose \emph{atoms} are the Hilbert spaces in $\X$. A quantum set consisting of a single atom is called \emph{atomic}. 

For convenience, we denote the elements of the index set $\At(\X)$ of $\X$ by the atoms of $\X$ themselves, hence $\At(\X)$ can be interpret as the set of atoms of $\X$. Thus, in some sense $\X$ is indexed by itself, just like ordinary sets can be regarded as indexed families indexed by themselves via the identity function. 
Since a quantum set $\X$ is formally a indexed family, it does not have elements in the usual sense. We shall use the notation $X\atomof\X$ to express that $X$ is an atom of $\X$. Thus, we have  $X\atomof\X$ if and only if $X\in\At(\X)$. Conversely,  to any ordinary set $M$ consisting of finite-dimensional Hilbert spaces we can associate a unique quantum set $\Q M$ whose set of atoms $\At(\Q M)$ consists of all Hilbert spaces $H$ in $M$ such that $\dim(H)\neq 0$. 

Any morphism in $\qRel$ is called a \emph{binary relation}. We emphasize that binary relations between quantum sets are not binary relations in the usual sense, i.e., subsets of the product of domain and codomain of the relation. However, binary relations between quantum sets turn out to be generalizations of binary relations between ordinary sets. Given our convention that the indices in the index set $\At(\X)$ of $\X$ are chosen to be the atoms of $\X$ itself, any binary relation $R:\X\to\Y$ between quantum sets is an assignment that to each atom $X\atomof\X$ and each atom $Y$ of $\Y$ assigns a subset $R(X,Y)$ of $L(X,Y)$. Given another binary relation $S:\Y\to\Z$, the $(X,Z)$-component of the composition $S\circ R$ is given by $(S\circ R)(X,Z)=\bigvee_{Y\atomof\Y}S(Y,Z)\cdot R(X,Y)$. The identity morphism on a quantum set $\X$ is denoted by $I_\X$ and is given by $I_\X(X,X')=\CC 1_X$ if $X=X'$ and $I_\X(X,X')=0$ otherwise.
The quantaloid structure of $\qRel$ can be described explicitly as follows. For binary relation $R,S:\X\to\Y$ we have   $R\leq S$ if and only if $R(X,Y)\leq S(X,Y)$ for each $X\atomof\X$ and each $Y\atomof\Y$. Equipped with this order, any homset of $\qRel$ becomes a complete lattice. The supremum $\bigvee_{i\in I}R_i$ of a collection $(R_i\colon i\in I)$ of relations $\X\to\Y$ is given by $\left(\bigvee_{i\in I}R_i\right)(X,Y)=\bigvee_{i\in I}R_i(X,Y)$ for each $X\atomof\X$ and each  $Y\atomof\Y$, where the supremum in the right-hand side is taken in the complete lattice of subspaces of $L(X,Y)$.


We can generalize the following set-theoretic notions to the quantum setting:
\begin{itemize}
    \item[(1)] A quantum set $\X$ is \emph{empty} if $\At(\X)=\emptyset$;
    \item[(2)] A quantum set $\X$ is a \emph{subset} of a quantum set $\Y$ if $\At(\X)\subseteq\At(\Y)$, in which case we write $\X\subseteq\Y$. 
    \item[(3)] The \emph{cartesian product} $\X\times\Y$ of two quantum sets $\X$ and $\Y$ is defined by $\At(\X\times\Y)=\{X\otimes Y\colon X\atomof\X,Y\atomof\Y\}$, where $X\otimes Y$ denotes the usual tensor product of the Hilbert spaces $X$ and $Y$. 
\end{itemize}
We denote the cartesian product of quantum sets by $\times$, because it is the noncommutative generalization of the usual product. However, it is not  a categorical product in any of the categories that we will introduce below. It is not uncommon to use the notation $\times$ for a non-categorical product: for instance, it is also used to denote the monoidal product in the category $\Rel$ of sets and binary relations.

To each ordinary set $S$ we can assign a quantum set $`S$ whose atoms are one-dimensional Hilbert spaces that are in a one-to-one correspondence with elements of $S$. This correspondence can be made precise as follows. For each $s\in S$, we define $\CC_s:=\ell^2(\{s\})$ with the convention that $\CC_s\neq \CC_t$ if $s\neq t$. Then $\At(`S)=\{\CC_s\colon s\in S\}$. Note that $`(S\times T)$ is isomorphic to $ (`S)\times(`T)$ as a quantum set.


It is well known that the category $\mathbf{FdHilb}$ of finite-dimensional Hilbert spaces and linear maps is a dagger compact category, where the dagger of a linear map $a$ is given by taking its  Hermitian adjoint $a^\dagger$.

Also $\mathbf{qRel}$ is a dagger compact category: for a relation $R\colon\X\to\Y$, we define $R^\dag\colon\Y\to\X$ by $R^\dag(Y,X)=\{a^\dag\colon a\in R(X,Y)\}$ for each $X\atomof\X$ and each $Y\atomof\Y$. The cartesian product $\times$ of quantum sets extends to a monoidal product that is defined on morphisms $R\colon\X\to\Y$ and $S\colon\W\to\Z$ by $(R\times S)(X\otimes W,Y\otimes Z)=R(X,Y)\otimes S(W,Z)$ for each $X\otimes W\atomof\X\times\W$ and each $Y\otimes Z\atomof\Y\times\Z$. The monoidal unit $\mathbf 1$ is given by the quantum set consisting of a single one-dimensional atom, typically denoted by $\CC$.

Let $H$ and $K$ be Hilbert spaces. For each linear operator $v \in L(H, K)$, write $v^* \in L(K^*, H^*)$ for the Banach space adjoint of $v$, defined by $v^*(\varphi) = \varphi \circ v$. For each subspace $V \leq L(H, K)$, write $V^* = \{v^* \colon v \in V\} \leq L(K^*, H^*)$. The \emph{dual} of a quantum set $\X$ is the quantum set $\X^*$ determined by $\At(\X^*)=\{X^*\colon X \atomof \X\}$. The \emph{dual} of a binary relation $R$ from $\X$ to $\Y$ is the binary relation $R^*$ from $\Y^*$ to $\X^*$ defined by $R^*(Y^*,X^*) = R(X, Y)^*$.
In $\qRel$, the associator $A$, the unitors $L$ and $R$, the symmetry $S$, the unit $H$ and the counit $E$ of the compact structure can be expressed in terms of the associator $\alpha$, the unitors $\lambda$ and $\rho$, the symmetry $\sigma$, and the unit $\eta$ and counit $\epsilon$ of the compact structure of $\mathbf{FdHilb}$. For instance, the nonzero components of $S_{\X\Y}\colon\X\times\Y\to\Y\times\X$ is given by $S_{\X\Y}(X\otimes Y,Y\otimes X)=\CC\sigma_{XY}$, and the nonzero components of $E_\X\colon\X\times\X^*\to\mathbf 1$ are given by $E_\X(X\otimes X^*,\CC)=\CC\epsilon_X$. 


The assignment $S\mapsto `S$ extends to a fully faithful functor $`(-)\colon\Rel\to\qRel$, which is defined on ordinary binary relations $r\colon S\to T$ for each $s\in S$ and each $t\in T$ by $(`r)(\CC_s,\CC_t)=L(\CC_s,\CC_t)$ if $(s,t)\in r$, and $(`r)(\CC_s,\CC_t)=0$ otherwise. Since $L(\CC_s,\CC_t)$ is one dimensional, it only has two subspaces, whence $`(-)$ is indeed fully faithful. Moreover, it preserves the dagger structure, and the inclusion order on homsets of $\Rel$. 

It is easy to verify that a function $f\colon X\to Y$ between ordinary sets is a binary relation such that $f^\dag\circ f\geq 1_X$ and $f\circ f^\dag\leq 1_Y$, where $f^\dag$ is the opposite relation of $f$. Hence, a relation $F\colon\X\to\Y$ between quantum sets is called a \emph{function} if it satisfies $F^\dag\circ F\geq I_\X$ and $F\circ F^\dag\leq I_\Y$. Examples of functions are the associator, unitors, and symmetry of $\qRel$. Another example of a function is provided by subsets $\Y$ of quantum sets $\X$, for which there is an \emph{inclusion function} $J_\Y^\X$  defined for each $Y\atomof\Y$ and each $X\atomof\X$ by $J_\Y^\X(Y,X)=\CC 1_Y$ if $Y=X$ and $J_\Y^\X(Y,X)=0$ otherwise. If it is clear that $\X$ is the ambient quantum set, we often write $J_\Y$ instead of $J_\Y^\X$. 

Given a binary relation $R:\X\to\Y$ and subsets $\Z\subseteq\X$ and $\W\subseteq\Y$, we define the \emph{restriction} $R|_\Z$ of $R$ to $\Z$ as the relation $R\circ J_\Z^\X$. The \emph{corestriction} $R|^\W$ of $R$ to $\W$ is defined as the relation $(J_\W^\Y)^\dag\circ R$. We have $(R|_\Z)|^\W=(R|^\W)|_\Z$, which we denote as $R|_\Z^\W$.

The wide subcategory of $\qRel$ of functions is denoted by $\qSet$, which is complete, cocomplete and symmetric monoidal closed with respect to the monoidal product $\times$. The monoidal unit, associator, unitors and symmetry are the same as for $\qRel$. 

A function $F\colon\X\to\Y$ is called \emph{injective} if $F^\dag\circ F=I_\X$ and \emph{surjective} if $F\circ F^\dag=I_\Y$. Any inclusion function is an injective map. The injective and surjective functions are precisely the respective monomorphisms and epimorphisms of $\qSet$. Functions that are both injective and surjective are called \emph{bijective}, and are precisely the isomorphisms of $\qSet$. The \emph{range} of a function $F\colon\X\to\Y$ is the quantum set $\ran F$ specified by $\At(\ran F)=\{Y\atomof\Y\colon F(X,Y)\neq 0\text{ for some }X\atomof\X\}$. We have $F=J_{\ran F}\circ \bar F$ for some unique surjective function $\bar F\colon\X\to\ran F$, which is defined by $\bar F(X,Y)=F(X,Y)$ for each $X\atomof\X$ and each $Y\atomof\ran F$. It follows that $F$ is surjective if and only if $\ran F=\Y$. 

The functor $`(-)\colon\Rel\to\qRel$ restricts and corestricts to a fully faithful functor $`(-)\colon\Set\to\qSet$. Furthermore, if we denote the category of von Neumann algebras and normal $*$-homomorphisms by $\mathbf{WStar}$, then there is a fully faithful functor $\ell^\infty\colon\qSet\to\mathbf{WStar}^\op$ that on objects is defined by $\X\mapsto\bigoplus_{X\atomof\X}L(X)$. The essential image of this functor is the category of \emph{hereditarily atomic} von Neumann algebras, i.e., von Neumann algebras that are isomorphic to some (possibly infinite) $\ell^\infty$-sum of matrix algebras. Also $\qRel$ can be shown to equivalent to a category of operator algebras \cite{Kornell20}, namely the category of hereditarily atomic von Neumann algebras and Weaver's quantum relations \cite{Weaver10}.

\subsection{Quantum posets}
The basic reference for this section is \cite{KLM22}.
    Let $\X$ be a quantum set. Then we call a binary relation $R\colon\X\to\X$ \emph{reflexive} if $I_\X\leq R$,
\emph{transitive} if $R\circ R\leq R$, and \emph{antisymmetric} if $R\wedge R^\dag\leq I_\X$.
   A pair $(\X,\qo)$ consisting of a quantum set $\X$ and a reflexive, transitive and antisymmetric relation $\qo$ on $\X$ is called a \emph{quantum poset}. The relation $\qo$ is called an \emph{order}. In order to improve the readability of expressions and calculations, we sometimes write parentheses around $\qo$, so we write $(\qo)$. 

\begin{example}
Let $\X$ be a quantum set. Then $I_\X$ is a quantum order on $\X$, which we call the \emph{trivial} order.
\end{example}

\begin{example}\label{ex:non-classical quantum poset}
Let $\H$ be the quantum set consisting of a single two-dimensional atom $H$. A  `non-classical' order on $\H$ is given by the relation $\qo$ on $\H$ specified by $
\qo(H,H)=\CC\begin{pmatrix}
1 & 0 \\
0 & 1
\end{pmatrix} +\CC \begin{pmatrix}
0 & 1 \\
0 & 0
\end{pmatrix}.$
Since $\H$ has only one atom $H$, $\qo$ is determined by $\qo(H,H)$. Thus, $(\H,\qo)$ is a quantum poset. 
\end{example}

  Let $(\X,\qo)$ be a quantum poset.  The relation $\qo^\dag$ is also an order, and is called the \emph{opposite} order on $\X$. We write $(\X,\qo)^\op=(\X,\qop)$, which is called the \emph{opposite} quantum poset of $(\X,\qo)$. Often, we just write $\X$ instead of $(\X,\qo)$ and $\X^\op$ instead of $(\X,\qop)$.
\begin{example}
    Let $(\H,\qo)$ be the quantum poset of the previous example. Then $\qop$ is specified by 
     $\qop(H,H)=\CC\begin{pmatrix}
1 & 0 \\
0 & 1
\end{pmatrix} +\CC \begin{pmatrix}
0 & 0 \\
1 & 0
\end{pmatrix}.$\end{example}

Let $(\X,\qo_\X)$ and $(\Y,\qo_\Y)$ be quantum posets. Then we say that a function $F\colon\X\to\Y$ is \emph{monotone} if $F\circ(\qo_\X)\leq (\qo_\Y)\circ F$.
Under the composition of functions between quantum sets, quantum posets and monotone functions form a category, which we call $\qPOS$, which is complete, cocomplete and monoidal closed under the monoidal product that is defined by $(\X,\qo_\X)\times(\Y,\qo_\Y)=(\X\times\Y,\qo_\X\times\qo_\Y)$. 
The monoidal unit is given by $(\mathbf 1, I_\mathbf 1)$. The components of the associator, unitors and symmetry are the components of the respective associator, unitors and symmetry of the underlying quantum sets.
We denote the evaluation morphism of $\qPOS$ by $\Eval_\sqsubseteq$, and the internal hom by $[\cdot,\cdot]_\sqsubseteq$. We call the isomorphisms of $\qPOS$ \emph{order isomorphisms}.

Let $(\X,\qo_\X)$ and $(\Y,\qo_\Y)$ be quantum posets. Then a monotone map $F\colon\X\to\Y$ is called an \emph{order embedding} if $(\qo_\X)=F^\dag\circ (\qo_\Y)\circ F$. The surjective order embeddings are precisely the order isomorphisms.
A subposet of a quantum poset $(\Y,\qo)$ consists of a subset $\X$ of $\Y$ equipped with the order $\qo|_\X^\X:=J_\X^\dag\circ \qo_\Y\circ J_\X$, to which we refer as the \emph{induced} order on $\X$. It follows that $J_\X\colon(\X,\qo|_\X^\X)\to(\Y,\qo)$ is an order embedding. 

Given a monotone map $F\colon(\X,\qo_\X)\to(\Y,\qo_\Y)$ between quantum posets, if we equip $\ran F\subseteq\Y$ with the relative order, then the unique surjective function $\bar F\colon\X\to\ran F$ such that $F=J_{\ran F}\circ \bar F$ is monotone. Hence, every monotone map can be written as the composition of a monotone surjective map and an order embedding.

If $(S,\sqsubseteq)$ is an ordinary poset, then $(`S,`\!\!\sqsubseteq)$ is a quantum poset, and vice versa; for example, the trivial order $\sqsubseteq$ on $S$ corresponds to $`(\sqsubseteq)\,=\, I_{`S}$, the trivial order on $`S$. Moreover, a monotone map $f$ between ordinary posets gives rise to a monotone function  $`f$ between the associated quantum posets, and vice versa. It follows that $`(-)$ extends to a fully faithful functor $\POS\to\qPOS$ defined on objects by $`(S,\sqsubseteq)=( `S,`{\sqsubseteq})$. If $2$ denotes the two-point set $\{0,1\}$, we write $\mathbf 2=`2$. If we equip $2$ with the order $\sqsubseteq$ defined by $0\sqsubset 1$, then we write $\qo_\mathbf 2=`{\sqsubseteq}$. Hence, $(\mathbf 2,\qo_\mathbf 2)=`(2,{\sqsubseteq})$.

Given a quantum set $\X$ and a quantum poset $(\Y,\qo)$, and given two functions $F,G\colon\X\to\Y$, we define $F\sqsubseteq_\Y G$ if $G\leq (\qo)\circ F$. This defines an order on $\qSet(\X,\Y)$ which is the quantum equivalent of the pointwise order of functions. We sometimes write $F\sqsubseteq G$ instead of $F\sqsubseteq_\Y G$.

\section{Monotone relations}

Let $(X,\sqsubseteq)$ and $(Y,\sqsubseteq)$ be posets. A binary relation $v\colon X\to Y$ is called a \emph{monotone relation} \cite{monotonerels} if $(x',y')\in v$, $x'\sqsubseteq x$ and $y\sqsubseteq y'$ implies $(x,y)\in v$. Under the usual composition of binary relations, posets and monotone relations form a category $\mathbf{RelPos}$, where the identity monotone relation $1_{(X,\sqsubseteq)}$ on the poset $(X,\sqsubseteq)$ is the binary relation $\sqsupseteq$. One can show that there are bijections between monotone relations $X\to Y$, monotone functions $X\times Y^\op\to 2$, and monotone functions $X\to \Dwn(Y)$, where $2$ is the two-point poset ${0,1}$ ordered by $0\sqsubset 1$, and $\Dwn(Y)$ is the poset of down-sets of $Y$ ordered by inclusion. In fact, the assignment $Y\mapsto \Dwn(Y)$ can be made into an endofunctor $\Dwn$ on $\POS$ that is a monad, and whose Kleisli category is isomorphic to $\mathbf{RelPos}$, reflected in the previous remark that a monotone relation $X\to Y$ corresponds to a monotone map $X\to \Dwn(Y)$, i.e., a Kleisli map. The importance of this monad lies in the fact that its Eilenberg-Moore category is precisely the category $\mathbf{Sup}$ of suplattices. 

A way to see that $\Dwn$ is the underlying endofunctor of a monad is the following. Just like we can embed $\Set$ into $\mathbf{Rel}$, we have an embedding $(-)_\diamond\colon\POS\to\mathbf{RelPos}$ that is the identify on objects, and that sends any monotone function $f\colon(X,\sqsubseteq)\to(Y,\sqsubseteq)$ to the monotone relation $f_\diamond=\{ (x,y)\colon y\sqsubseteq f(x)\}$.

Moreover, just like the covariant power set functor extends to a functor $\mathbf{Rel}\to\Set$ that is the right adjoint of the embedding $\Set\to\mathbf{Rel}$, the assignment $X\mapsto \Dwn(X)$ extends to a functor $\mathbf{RelPos}\to\POS$ that is the right adjoint of $(-)_\diamond$. The monad induced by this adjunction is precisely the down-set monad on $\POS$. 
Hence, in order to define quantum suplattices, we will have to find the quantum generalization of the down-set monad, and in order to find this quantum down-set monad, we have to find a quantum generalization of monotone relations.

\begin{definition}
Let $(\X,\qo_\X)$ and $(\Y,\qo_\Y)$ be quantum posets. A binary relation $V\colon\X\to\Y$ is called a \emph{monotone relation} $(\X,\qo_\X)\to(\Y,\qo_\Y)$ if it satisfies one of the following two equivalent conditions (hence both):
\begin{itemize}
    \item[(1)] $(\qop_\Y)\circ V\leq V$ and $V\circ (\qop_\X)\leq V$.
    \item[(2)] $(\qop_\Y)\circ V=V=V\circ (\qop_\X)$. 
\end{itemize}
\end{definition}
The equivalence of both conditions follows from the reflexivity of orders.

\begin{lemma}
Let $V\colon(\X,\qo_\X)\to(\Y,\qo_\Y)$ and $W\colon(\Y,\qo_\Y)\to(\Z,\qo_\Z)$ be monotone relations between quantum posets. Then $W\circ V\colon(\X,\qo_\X)\to(\Z,\qo_\Z)$ is a monotone relation.
\end{lemma}

\begin{lemma}
Let $(\X,\qo_\X)$ be a quantum poset. Then $\qop_\X\colon(\X,\qo_\X)\to(\X,\qo_\X)$ is a monotone relation such that for each quantum poset $(\Y,\qo_\Y)$ and all monotone relations $V\colon(\X,\qo_\X)\to(\Y,\qo_\Y)$ and $W\colon(\Y,\qo_\Y)\to(\X,\qo_\X)$ we have $V\circ (\qop_\X)=V$ and $(\qop_\X)\circ W=W$.
\end{lemma}

It follows from the previous two lemmas that the following definition is sound:
\begin{definition}
We define the category of quantum posets and monotone relations by $\mathbf{qRelPos}$. The identity monotone relation on a quantum poset $(\X,\qo_\X)$ is $\qop_\X$, which we often denote by  $I_{(\X,\qo)}$.
\end{definition}

\begin{lemma}
    There is a fully faithful functor $`(-)\colon\mathbf{RelPos}\to\mathbf{qRelPos}$ that sends any poset $(S,\sqsubseteq)$ to $(`S,`{\sqsubseteq})$ and  any monotone relation $v\colon(S,\sqsubseteq)\to(T,\sqsubseteq)$ to $`v$.
\end{lemma}

\begin{lemma}\label{lem:embedding qPOS into qRelPos}
There is a faithful functor $(-)_\diamond\colon\qPOS\to\mathbf{qRelPos}$ which is the identity on objects, and which acts on monotone maps $F\colon(\X,\qo_\X)\to(\Y,\qo_\Y)$ by
$F_\diamond:=(\qop_\Y)\circ F$.
\end{lemma}
The functor in the previous lemma is an extension of the functor $(-)_\diamond\colon\POS\to\mathbf{RelPos}$ mentioned above, in the sense that $\POS\xrightarrow{`(-)}\qPOS\xrightarrow{(-)_\diamond}\qRelPos$ equals $\POS\xrightarrow{(-)_\diamond}\RelPos\xrightarrow{`(-)}\qRelPos$.


\begin{definition}
Let $(\X,\qo_\X)$ be a quantum poset. Then we define $(\X,\qo_\X)^*$ to be the quantum poset $(\X^*,\qo_\X^*)$. Sometimes, we write $\X^*$ instead of $(\X,\qo_\X)^*$.
\end{definition}

Since the operation of taking daggers in dagger compact categories commutes with the operation of taking duals, we obtain the following lemma:
\begin{lemma}
    Let $(\X,\qo)$ be a quantum poset. Then $(\X^*)^\op=(\X^\op)^*$.
\end{lemma}

\begin{theorem}\label{thm:qRelPos is compact}
The category $\mathbf{qRelPos}$ is compact closed: for each quantum poset $(\X,\qo)$, the unit $H_{(\X,\qo)}\colon(1,I_\mathbf 1)\to (\X,\qo)^*\times(\X,\qo)$ and the counit $E_{(\X,\qo)}\colon(\X,\qo)\times(\X,\qo)^*\to(\mathbf 1,I_\mathbf 1)$ are given by $(\qop^*\times \qop)\circ H_\X$ and  $E_{\X}\circ (\qop\times \qop^*)$, respectively, where $H_\X$ and $E_\X$ denote the usual unit and counit of $\mathbf{qRel}$. The associator, unitors and symmetry are obtained by applying the functor $(-)_\diamond$ to the usual associator, unitors and symmetry in $\qPOS$.
\end{theorem}

\section{The quantum down-set monad}
In this section, we introduce the quantum down-set monad. Its construction shares similarities with the construction of the quantum power set in \cite{KLM22}. This construction yields a quantum poset, but the construction of this order seems to be a bit ad hoc. The framework of monotone relations seems to be more appropriate for the construction of an ordered object. We will see that the quantum down-set monad by means of monotone relations is ordered in a natural way. When we apply the monad to a trivially ordered quantum set, then we obtain the quantum power set of this quantum set.

\begin{definition}
 We define the quantum poset $\qDwn(\X,\qo)$ of down-sets of a quantum poset $(\X,\qo)$ to be the  internal hom in $\qPOS$ from $(\X,\qo)^*$ to $(\mathbf 2,\qo_{\mathbf 2})$, i.e., $\qDwn(\X,\qo):=[(\X,\qo)^*,(\mathbf 2,\qo_\mathbf 2)]_\sqsubseteq$. We will denote its underlying quantum set by $\D(\X,\qo)$. The order on $\D(\X,\qo)$ is denoted by $\qsubseteq_{(\X,\qo)}$, so $\qDwn(\X,\qo)=(\D(\X,\qo),\qsubseteq_{(\X,\qo)})$.
\end{definition}
Note that the order $\qsubseteq$ on $\qDwn(\X)$ is a boldface symbol to distinguish it from the inclusion order $\subseteq$ between ordinary sets.
We will prove that the assignment of objects $(\X,\qo)\mapsto\qDwn(\X,\qo)$ extends to a monad on $\qPOS$ by showing that the functor $(-)_\diamond\colon\qPos\to\qRelPos$ has a right adjoint; the monad is then induced by this adjunction. The right adjoint also sends objects $(\X,\qo)$ to $\qDwn(\X,\qo)$. Nevertheless, it is useful to make a distinction in the notation between monad and right adjoint, hence we will denote the right adjoint by $\qDwn'$. The first step of showing the existence of $\qDwn$ is the following lemma, for which we note that we have embeddings $1\to 2$ which map $*\in 1$ to either $0\in 2$ or $1\in 2$. We denote the respective maps by $0$ and $1$ as well. As a consequence, we have functions $`0,`1\colon\mathbf 1\to\mathbf 2$. 

\begin{lemma}\label{lem:qRelPos1isqPos2}
Let $(\X,\qo_\X)$ be a quantum poset. Then the bijection $\qSet(\X,\mathbf 2)\to \mathbf{qRel}(\X,\mathbf 1)$, $F\mapsto `1^\dag\circ F$ in \cite[Theorem B.8]{Kornell18}  restricts and corestricts to a bijection \[\qPOS((\X,\qo_\X),(\mathbf 2,\qo_\mathbf 2))\to\qRelPos((\X,\qo_\X),(\mathbf 1,I_{\mathbf 1})).\]
\end{lemma}

The counit of the adjunction that yields the ordinary down-set monad is the inverse membership relation $\ni$. Lemma \ref{lem:qRelPos1isqPos2} assures the existence of the quantum equivalent of this counit, which we will denote with a boldface symbol $\qni$.

\begin{lemma}\label{lem:counit downset monad}
For any quantum poset $(\X,\qo)$ there is a unique monotone relation $\qni_{(\X,\qo)}\colon\qDwn(\X,\qo)\to(\X,\qo)$ such that $
`1^\dag\circ\Eval_{\sqsubseteq}=E_{(\X,\qo)}\circ (\qni_{(\X,\qo)}\times I_{(\X,\qo)^*}).$
\end{lemma}

\begin{theorem}\label{thm:downset}
There is a functor $\qDwn'\colon\qRelPos\to\qPOS$ whose action on objects is given by $(\X,\qo)\mapsto\D(\X,\qo)$, and which is right adjoint to the functor $(-)_\diamond\colon\qPOS\to\qRelPos$. The $(\X,\qo)$-component of the counit $\qni$ of this adjunction is the monotone relation constructed in Lemma \ref{lem:counit downset monad}. The unit of the adjunction is denoted by $\qdown$.  Its $(\X,\qo)$-component is an order embedding that is the unique monotone function $(\X,\qo)\to\D(\X,\qo)$ such that
$\qni_{(\X,\qo)}\circ\qdown_{(\X,\qo)}=I_{(\X,\qo)}.$
\end{theorem}

\begin{definition}
    We define the \emph{quantum down-set monad} $\qDwn$ to be the monad induced by the adjunction $(-)_\diamond\dashv\qDwn'$, so $\qDwn=\qDwn'\circ(-)_\diamond$. We denote its multiplication by $\qun$ and its unit by $\qdown$.
\end{definition}
Note that the multiplication $\qun$ is a boldfaced version of the usual union $\bigcup$ of ordinary sets. 

\section{Opposite quantum posets and upper sets}

Let $X$ be an ordinary poset. Then the complementation operator provides a bijection between the set $D(X)$ of down sets of $X$ and the set $U(X)$ of upper sets of $X$. Both $D(X)$ and $U(X)$ can be extended to endofunctors on $\POS$, for which we have to order the former by inclusion and the latter by containment. Then, writing $\mathrm{Dwn}(X)=(D(X),\subseteq)$ and $\mathrm{Up}(X)=(U(X),\supseteq)$, the bijection extends to an  order isomorphism $\mathrm{Dwn}(X)\to\mathrm{Up}(X)$.  In the quantum world, we can obtain a similar order isomorphism by showing that we can also construct a different right adjoint of $(-)_\diamond$ in terms of upper sets, which, by the uniqueness of right adjoints up to natural isomorphism, should be naturally isomorphic to $\qDwn'$. This natural isomorphism is precisely the operation of taking complements. Before we construct this right adjoint in terms of upper sets, we first have to extend the operation of taking opposite quantum posets to an endofunctor on $\qPOS$.

\begin{lemma}\label{lem:functions and the opposite order}
    There is an endofunctor $(-)^\op\colon\qPOS\to\qPOS$, defined on objects by $(\X,\qo_\X)\mapsto (\X,\qop_\X)$ and which maps any monotone map $F\colon(\X,\qo_
    \X)\to(\Y,\qo_\Y)$ to the monotone map $F\colon(\X,\qop_\X)\to(\Y,\qop_\Y)$. This functor $(-)^\op$ is involutory, i.e., $(-)^{\op\op}=1_{\qPOS}$, hence an isomorphism of categories.
\end{lemma}


\begin{proposition}\label{prop:XopYopisXYop}
Let $(\X,\qo_\X)$ and $(\Y,\qo_\Y)$ be quantum posets. Then $[\X^\op,\Y^\op]_\sqsubseteq=[\X,\Y]^\op_\sqsubseteq.$
\end{proposition}

\begin{definition}
    Let $(\X,\qo)$ be a quantum poset. Then we define the \emph{quantum poset of upsets of} $(\X,R)$ as the quantum poset $\qUp(\X,\qo):=[(\X,\qo)^*,(\mathbf 2,\qop_\mathbf 2)]_\sqsubseteq$. We denote its underlying quantum set by $\U(\X,\qo)$.
\end{definition}
The previous proposition yields $\qUp(\X,\qo)=[\X^*,\mathbf 2^\op]_\sqsubseteq=[(\X^\op)^*,\mathbf 2]_\sqsubseteq^\op=
\qDwn(\X,\qop)^\op$, whence $\U(\X,\qo)=\D(\X,\qop)$.
The other right adjoint is obtained by constructing a different counit, namely the inverse non-membership relation. This is done by taking $\mathbf 2^\op$ instead of $\mathbf 2$ and $`0^\dag\circ F$ instead of $`1^\dag\circ F$ in Lemma \ref{lem:qRelPos1isqPos2}: 
Then given a quantum poset $(\X,\qo)$, we have that $\qnotni_{(\X,\qo)}$ is the unique monotone relation $\qUp(\X,\qo)\to(\X,\qo)$ such that $
`0^\dag\circ\Eval_{\sqsubseteq}=E_{(\X,\qo)}\circ (\qnotni_{(\X,\qo)}\times I_{(\X,\qo)^*}).$

\begin{theorem}\label{thm:upset}
There is a functor $\qUp'\colon\mathbf{qRelPos}\to\qPOS$ whose action on objects is given by $(\X,\qo)\mapsto\qUp(\X,\qo)$ that is right adjoint to the functor $(-)_\diamond\colon\qPOS\to\mathbf{qRelPos}$. Its counit is denoted by $\qnotni$, and its  $(\X,\qo)$-component of the counit is the monotone relation $\qnotni_{(\X,\qo)}$.
\end{theorem}

\begin{definition}
    The monad $\qUp'\circ(-)_\diamond$ on $\qPOS$ is denoted by $\qUp$. Note that its action on objects agrees with the assignment $(\X, \qo)\mapsto\qUp(\X,\qo)$.
\end{definition}

\begin{corollary}\label{cor:complement}
There is a natural isomorphism $C\colon\qDwn'\to \qUp'$ between functors $\mathbf{qRelPos}\to\qPOS$ that induces a natural isomorphism between the endofunctors $\qDwn\to \qUp$ on $\qPOS$ that we also denote by $C$. In particular, for a quantum poset $(\X,\qo)$, the $(\X,\qo)$-component of $C$ is an order isomorphism $C_{(\X,\qo)}\colon\qDwn(\X,\qo)\to\qUp(\X,\qo)$ that is natural in $(\X,\qo)$ such that $\qni_{(\X,\qo)}\circ C_{(\X,\qo)}=\qnotni_{(\X,\qo)}$. 
\end{corollary}

\section{The quantum power set}\label{sec:quantum power set}

In \cite{KLM22}, the quantum power set monad $\P$ was introduced. We sketch how to derive this monad from the quantum down-set monad. We have a clear categorical embedding $\qSet\to\qPOS$ acting on objects by $\X\mapsto(\X,I_\X)$, which is left adjoint to a forgetful functor acting on objects by $(\X,\qo)\mapsto \X$. We can now define the quantum power set $\P(\X)$ of a quantum poset $\X$ as the quantum set $\D(\X,I_\X)$, which can be made into a functor by composing $\D$ with the embedding of $\qSet$ into $\qPOS$. By playing with the composition of adjunctions, we can equip $\P$ with the structure of a monad such as in \cite{KLM22}. The ordered quantum power set $\qPow(\X)$ can be obtained as $\qDwn(\X,I_\X)$.  We note that $\P(\X)$ is also the underlying quantum set of $\qUp(\X,I_\X)$. The natural isomorphism $C$ from Corollary \ref{cor:complement} now yields a bijection $C_{(\X,I_\X)}:\P(\X)\to\P(\X)$ that is an order isomorphism $\qPow(\X)\to\qPow(\X)^\op$, and which can be regarded as the quantum analog of the complement operator on the power set. For simplicity, we write $C_\X$ instead of $C_{(\X,I_\X)}$.

\section{Galois connections}

\begin{definition}
Let $(\X,\qo_\X)$ and $(\Y,\qo_\Y)$ be quantum posets and let $F\colon\X\to\Y$ and $G\colon\Y\to\X$ be functions. If 
$(\qo_\Y)\circ F=G^\dag\circ (\qo_\X),$
we say that $(F,G)$ forms a \emph{Galois connection}, or that $F$ is the \emph{lower adjoint} of $G$, or that $G$ is the \emph{upper adjoint} of $F$. 
\end{definition}
This definition is a generalization of the usual definition of a Galois connection between ordinary posets: the $`(-)$ functor maps ordinary Galois connections to Galois connections in the sense of the definition above. Similarly as in the classical case, the lower adjoint in a Galois connection determines the upper adjoint and \emph{vice versa}. The next theorem provides alternative characterizations of Galois connections, which might look more familiar to the classical case:
\begin{theorem}\label{thm:Galois}
Let $(\X,\qo_\X)$ and $(\Y,\qo_\Y)$ be quantum posets and let $F\colon\X\to\Y$ and $G\colon\Y\to\X$ be monotone functions. Then the following statements are equivalent:
\begin{itemize}
    \item[(1)]$F$ is the lower adjoint of $G\colon\Y\to \X$.
    \item[(2)] 
    $F\circ K\sqsubseteq_\Y M\iff K\sqsubseteq_\X G\circ M$ for any quantum set $\Z$ and  functions $K\colon\Z\to\X$ and $M\colon\Z\to\Y$.
\item[(3)] We have $I_\X\sqsubseteq_\X G\circ F$ and $ F\circ G\sqsubseteq_\Y I_\Y$.
\end{itemize}
\end{theorem}

\begin{proof}
    We start by showing that (1) implies (2). So assume that $F$ is the lower adjoint of $G$, so $(\qo_\Y)\circ F=G^\dag\circ (\qo_\X)$. Let $\Z$ be a quantum set and let $K:\Z\to\X$ and $M:\Z\to\Y$ be functions. 
 Assume $F\circ K\sqsubseteq_\Y M$. By definition of $\sqsubseteq_\Y$, we have $M\leq (\qo_\Y)\circ F\circ K$, so $M\leq G^\dag\circ \qo_\X\circ K$, hence $G\circ M\leq G\circ G^\dag\circ \qo_\X\circ K\leq (\qo_\X)\circ K$ since $G$ is a function. Hence $K\sqsubseteq_\X G\circ M$.

Conversely if $K\sqsubseteq_\X G\circ M$, we have $G\circ M\leq (\qo_\X)\circ K$ by definition of a function,
$M\leq G^\dag\circ G\circ M\leq G^\dag\circ \qo_\X\circ K=(\qo_\Y)\circ F\circ K$, so $F\circ K\sqsubseteq_\Y M$.

We show that (2) implies (3). Since $F\sqsubseteq_\Y F$,  condition (2) yields $I_\X\sqsubseteq_\X G\circ F$ if we choose $\Z=\X$, $K=I_\X$ and $M=F$. Since $G\sqsubseteq_\X G$, we obtain $F\circ G\sqsubseteq_\Y I_\Y$ by choosing $\Z=\Y$, $K=G$ and $M=I_\Y$.

We proceed by showing that (3) implies (1). Since $I_\X\sqsubseteq_\X G\circ F$, we have $G\circ F\leq (\qo_\X)$. Then $F\leq G^\dag\circ G\circ F\leq G^\dag\circ (\qo_\X)$. By Lemma \ref{lem:functions and the opposite order}, $G:\Y^\op\to\X^\op$ is also monotone, so $G\circ (\qop_\Y)\leq (\qop_\X)\circ G$, which implies $(\qo_\Y)\circ G^\dag\leq G^\dag\circ (\qo_\X)$ after taking daggers. Hence, $(\qo_\Y)\circ F\leq (\qo_\Y)\circ G^\dag\circ \qo_\X\leq G^\dag\circ( \qo_\X\circ \qo_\X)=G^\dag\circ (\qo_\X)$.

Since $F\circ G\sqsubseteq_\Y I_\Y$, we have $I_\Y\leq (\qo_\Y)\circ F\circ G$. Hence, $G^\dag\leq (\qo_\Y)\circ F\circ G\circ G^\dag\leq (\qo_\Y)\circ F$. Since $F$ is monotone, we have $F\circ (\qo_\X)\leq (\qo_\Y)\circ F$, whence $G^\dag\circ (\qo_\X)\leq (\qo_\Y)\circ F\circ (\qo_\X)\leq (\qo_\Y\circ \qo_\Y)\circ F=(\qo_\Y)\circ F$.
We conclude that $(\qo_\Y)\circ F=G^\dag\circ (\qo_\X)$, so $F$ is the lower adjoint of $G$. 
\end{proof}

The next example is the quantum version of the statement that the direct image and the preimage of a function form a Galois connection:
\begin{example}
$\qPow(F)\colon\qPow(\X)\to\qPow(\Y)$ is the lower Galois adjoint of $\qPow(F^\dag)$  for any function $F\colon\X\to\Y$.
\end{example}

Also the notion of closure operators can be generalized to the quantum setting:

\begin{definition}
Let $(\X,\qo)$ be a quantum poset. Then we call a monotone function $C\colon\X\to\X$ a \emph{closure operator} on $\X$ if $I_\X\sqsubseteq C$ and $C\circ C=C$.
\end{definition}

Just as in the classical case, there is a relation between Galois connections and closure operators:

\begin{theorem}\label{thm:closure operator vs Galois connection}
Let $(\X,\qo_\X)$ and $(\Y,\qo_\Y)$ be quantum posets and let $F\colon\X\to\Y$ be a monotone function that is the lower adjoint of a monotone function $G\colon\Y\to\X$. Then $C:=G\circ F$ is a closure operator on $\X$.
Conversely, let $C$ be a closure operator on a quantum poset $(\X,\qo_\X)$. Let $\Y=\ran C$, and let $\qo_\Y$ be the induced order on $\Y$, i.e., $\qo_\Y=\qo_\X|_\Y^\Y$. Then the unique surjective monotone function $\bar C\colon\X\to\Y$ such that $C=J_\Y\circ\bar C$ is is the lower adjoint of the order embedding $J_\Y\colon\Y\to\X$. In particular, we have $\bar C\circ J_\Y=I_\Y$.
\end{theorem}

The quantum version of the operation $A\mapsto\down A$ on a power set $\Pow(X)$ is a closure operator:
\begin{example}\label{ex:downward closure operator} 
    Let $\qo$ be an order on a quantum set $\X$. Then $\qPow(\qop)$ is a closure operator on $\qPow(\X)$. Its range equipped with the relative order equals $\qDwn(\X,\qo)$.   
\end{example}

\section{Quantum suplattices}
Classically, a poset $X$ is a complete lattice if its canonical embedding $X\to\Dwn(X)$, $x\mapsto\down x$ into its poset of down-sets has a lower Galois adjoint. We will use this fact in order to define quantum suplattices. 
\begin{definition}
Let $(\X,\qo_\X)$ be a quantum poset, and let $\qdown_{(\X,\qo_\X)}\colon(\X,\qo_\X)\to\qDwn(\X,\qo_\X)$ be the order embedding of $\X$ into the quantum poset of down-sets of $\X$. Then we say that $(\X,\qo_\X)$ is a quantum suplattice if $\qdown_{(\X,\qo_\X)}$ has a lower adjoint $\qsup_{(\X,\qo_\X)}\colon\qDwn(\X,\qo_\X)\to(\X,\qo_\X)$. Moreover, if $(\Y,\qo_\Y)$ is another quantum suplattice, then we say that a function $K\colon\X\to\Y$ is a \emph{homomorphism of quantum suplattices} if  $K\circ \qsup_{(\X,\qo_\X)}=\qsup_{(\Y,\qo_\Y)}\circ\D(K).$ We denote the category of quantum suplattices and homomorphisms of quantum suplattices by $\qSup$.
\end{definition}
We will show that quantum posets of down-sets form the primary examples of quantum suplattices. In the classical case, this is completely obvious; one just need to observe that a union of down-sets is a down-set. However, in the quantum case,
the proof is nontrivial. We need one crucial lemma.
\begin{lemma}\label{lem:Qorderidentity}
Let $(\X,\qo)$ be a quantum poset. Then the inverse inclusion order $\qsupseteq$ on $\D(\X)$ is the largest binary relation $T$ on $\D(\X)$ such that $
(\qni_{(\X,\qo)})\circ T\leq(\qni_{(\X,\qo)})$.
\end{lemma}

\begin{theorem}\label{thm:quantum down sets form quantum suplattice}
Let $(\X,\qo)$ be a quantum poset. Then $\qDwn(\X,\qo)$ is a quantum suplattice. More specifically, the lower Galois adjoint $\qsup_{\D(\X,\qo)}$ of the canonical embedding $\qdown_{\qDwn(\X,\qo)}\colon\qDwn(\X,\qo)\to\qDwn^2(\X,\qo)$ is given by the $(\X,\qo)$-component of $\qun_{(\X,\qo)}$ of the multiplication $\qun\colon\qDwn^2\to\qDwn$ of the $\qDwn$-monad on $\qPos$.  
\end{theorem}
\begin{proof}For simplicity, we write $\qdown$ instead of $\qdown_{\qDwn(\X,\qo)}$ and $\qun$ instead of $\qun_{(\X,\qo)}$. 
Consider the following two diagrams:
\[\begin{tikzcd}
\qDwn(\X)\ar{d}[swap]{\qdown}\ar{dr}{I_{\qDwn(
\X)}} &    && \qDwn(\X)\ar{d}[swap]{\qdown}\ar{ddr}{\qni_{(\X,\qo)}} & \\
\qDwn^2(\X)\ar{r}[swap]{\qni_{\qDwn(\X,\qo)}}\ar{d}[swap]{\qDwn'(\qni_{(\X,\qo)})} & \qDwn(\X)\ar{d}{\qni_{(\X,\qo)}} && \qDwn^2(\X)\ar{d}[swap]{\qun} &  \\
\qDwn(\X)\ar{r}[swap]{\qni_{(\X,\qo)}} & (\X,\qo) && \qDwn(\X)\ar{r}[swap]{\qni_{(\X,\qo)}} & (\X,\qo)
\end{tikzcd}\]
The triangle in the left diagram commutes by Theorem \ref{thm:downset}. The square commutes by naturality of $\qni$ and since $\qDwn'$ equals $\qDwn$ on objects. The right diagram commutes since  $\qun_{(\X,\qo)}=\qDwn(\qni_{(\X,\qo)})$, so it is the outside of the left diagram. By the universal property of $\ni$ as the counit of the adjunction $(-)_\diamond\dashv\qDwn$, the unique monotone function $K\colon\qDwn(\X)\to\qDwn(\X)$ such that $\qni_{(\X,\qo)}\circ K=\qni_{(\X,\qo)}$ is the identity function on $\D(\X)$.  Hence, we conclude that $\qun\circ \qdown=I_{\D(\X)}$. Note that the right-hand side is not the same as the identity monotone relation $I_{\qDwn(\X)}$, which is $\qsupseteq$, even though $\D(\X)$ is the underlying quantum set of $\qDwn(\X)$. 

As a consequence, $\qun$ is an epimorphism in $\qSet$, hence it is surjective, so $\qun\circ \qun^\dag=I_{\D(\X)}$. Then, using the naturality of $\qni$, we obtain
\[\qni_{(\X,\qo)}=\qni_{(\X,\qo)}\circ \qun\circ \qun^\dag=\qni_{(\X,\qo)}\circ\qDwn'(\qni_{(\X,\qo)})\circ \qun^\dag=\qni_{(\X,\qo)}\circ \qni_{\qDwn(\X,\qo)}\circ \qun^\dag,\]
which, by Lemma \ref{lem:Qorderidentity}, implies
$\qni_{\qDwn(\X)}\circ \qun^\dag\leq (\qsupseteq)=I_{\qDwn(\X)}$. Since also $I_{\qDwn(\X)}=\qni_{\D(\X,\qo)}\circ\qdown_{\D(\X,\qo)}$ (cf. Theorem \ref{thm:downset}), we obtain $\qni_{\qDwn(\X)}\circ \qun^\dag\leq \qni_{\qDwn(\X)}\circ\qdown$. Then, since $\qdown$ is a function, we find 
$\qni_{\qDwn(\X)}\circ \qun^\dag\circ \left(\qdown\right)^\dag\leq \qni_{\qDwn(\X)}\circ\qdown\circ \left(\qdown\right)^\dag\leq \qni_{\qDwn(\X)}$. Again applying  Lemma \ref{lem:Qorderidentity} yields $\qun^\dag\circ \left(\qdown\right)^\dag\leq(\qsupseteq)$, whence $(\qsubseteq)\circ \qdown\circ \qun \leq (\qsubseteq)\circ(\qsupseteq)^\dag=(\qsubseteq)\circ(\qsubseteq)\leq(\qsubseteq)=(\qsubseteq)\circ I_{\D^2(\X)}$, which expresses that $I_{\D^2(\X)}\leq \qun\circ \qdown$.
It now follows from Theorem \ref{thm:Galois} that $\qun$ is the lower Galois adjoint $\qsup_{\qDwn(\X)}$ of $\qdown$, hence $\qDwn(\X)$ is indeed a quantum suplattice.
\end{proof}

We can now formulate the quantum versions of two classical theorems on suplattices:
\begin{theorem}
    Let $(\X,\qo_\X)$ and $(\Y,\qo_\Y)$ be quantum suplattices and let $F\colon\X\to\Y$ be monotone. Then $F$ is a homomorphism of quantum suplattices if and only if it has an upper Galois adjoint.
\end{theorem}

\begin{theorem}
    The Eilenberg-Moore category of $\qDwn$ is equivalent to $\qSup$.
\end{theorem}

\section{Quantum inflattices}
We can define an inflattice to be the opposite of a suplattice. In the quantum world, we follow the same path to define quantum inflattices:
\begin{definition}
    A quantum poset $(\X,\qo)$ is called a \emph{quantum inflattice} if $(\X,\qop)$ is a quantum suplattice.
\end{definition}
Classically, a poset is a suplattice if and only if it is an inflattice. The same is true in the quantum case. In order to prove this, a first step is showing a different characterization of quantum suplattices, which is analogue to the observation that any ordinary poset $X$ is a suplattice if and only if the embedding $X\to \mathrm{Pow}(X)$, $x\mapsto\down x$ has a lower Galois adjoint. The analogue embedding $D_{(\X,\qo)}$ of a quantum poset $(\X,\qo)$ into $\qPow(\X)$ is given by $J\circ\qdown_\X$, where $J\colon\qDwn(\X,\qo)\to \qPow(\X)$ is the order embedding from Example \ref{ex:downward closure operator}. 
\begin{lemma}\label{lem:alt definition qsup}
    A quantum poset $(\X,\qo)$ is a quantum suplattice if and only if $D_{(\X,\qo)}\colon(\X,\qo)\to\qPow(\X)$ has a lower Galois adjoint $F_{(\X,\qo)}$.
\end{lemma}
Classically, for any ordinary set $X$ and any subset $\A\subseteq\Pow(X)$, the intersection $\bigcap\A$ of all subsets in $\A$ is given by $X\setminus\bigcup_{A\in\A}(X\setminus A)$, so by using the union operator and the complement operator. Since for any quantum set $\X$ its quantum power set $\qPow(\X)=\qDwn(\X,I_\X)$ is a quantum suplattice (cf. Theorem \ref{thm:quantum down sets form quantum suplattice}), and similar to the complementation operator on ordinary power sets, we have an order isomorphism $C_\X\colon\qPow(\X)\to\qPow(\X)^\op$, we can prove:
\begin{lemma}
   The quantum power set $\qPow(\X)$ of any quantum set $\X$ is a quantum inflattice.
\end{lemma}
Then, given a quantum suplattice $(\X,\qo)$ with embedding $D_{(\X,\qo)}\colon(\X,\qo)\to\qPow(\X)$ that is the upper adjoint of $F$, and denoting the lower adjoint of $D_{\qPow(\X)^\op}\colon\qPow(\X)^\op\to\qPow(\P(\X))$ by $N$, we can show that $F\circ N\circ\qPow(D_{(\X,\qo)})$ is the lower Galois adjoint of $D_{(\X,\qop)}\colon(\X,\qop)\to\qPow(\X)$, proving:
\begin{theorem}\label{thm:quantum suplattice if quantum inflattice}
Any quantum suplattice $(\X,\qo)$ is a quantum inflattice.
\end{theorem}

\section{Enrichment}\label{sec:enrichment}
It was shown in \cite{KLM22} that the pointwise order of functions between quantum sets induces a $\POS$-enrichment of $\qPOS$. Similarly, in \cite{KLM20} it was shown that the category $\qCPO$ of quantum cpos is enriched over the category $\CPO$ of cpos, i.e., posets for which any monotonically increasing sequence has a supremum. We can enrich $\qSup$ over $\Sup$ in a similar way. We first have to quantize the supremum of collections of functions into a quantum suplattice. 
\begin{definition}
    Let $\X$ be a quantum set and $(\Y,\qo)$ a quantum poset. Let $\kK$ be a subset of $\qSet(\X,\Y)$. Then a function $F\colon\X\to\Y$ is called the \emph{limit} of $\kK$, denoted by $F=\lim\kK$ 
    if $(\qo)\circ F=\bigwedge_{K\in\kK}(\qo)\circ K$.
    \end{definition}
A quantum cpo is a quantum poset $(\Y,\qo)$ for which the limit of any countable chain in $\qSet(\X,\Y)$ exists for any quantum set $\X$, so the above definition is a  generalization of the concept of limits in \cite{KLM22}.

Let $\X$ be a quantum set and let $(\Y,\qo)$ be a quantum poset. Let $\kK\subseteq\qSet(\X,\Y)$. If $\lim\kK$ exists, then it is the supremum of $\kK$ in $\qSet(\X,\Y)$ ordered by $\sqsubseteq$. The converse does not always hold. If $\Y$ is a quantum suplattice, then $\lim\kK$ always exists. We can prove this by first showing the case $\Y=\mathbf 2$, then the case $\Y=\qPow(\X)$, which equals $[\X^*,\mathbf 2]_\sqsubseteq$, and finally the general case using Lemma \ref{lem:alt definition qsup}. Moreover, one can show that $\lim\kK$ is a homomorphism of quantum suplattices if $\X$ is also a quantum suplattice and all functions in $\kK$ are homomorphisms of quantum suplattices. Finally, one can show that composition with homomorphisms of quantum suplattices preserves the operation of taking limits, yielding:  \begin{theorem}
$\qSup$ is enriched over $\Sup$.  
\end{theorem}

\section{Fixpoints}
Let $\X$ be a quantum set and $\Y$ a quantum suplattice. Since $\qSet(\X,\Y)$ is a complete lattice, it follows from the Knaster-Tarski fixpoint theorem that:
\begin{proposition}
    Let $F\colon\Y\to\Y$ be a monotone map on a quantum suplattice $\Y$. Then for each quantum set $\X$, the set of all functions $K\colon\X\to\Y$ such that $F\circ K=K$ is a complete lattice. 
\end{proposition}
Categorically, a generalized fixpoint of an endomorphism $f\colon Y\to Y$ in a given category $\mathbf C$ can be defined as a monomorphism $m\colon X\to Y$ for some object $X$ such that $f\circ m=m$. This leads to:
\begin{definition}
    Let $(\Y,\qo)$ be a quantum poset, and let $F\colon\Y\to\Y$ be a monotone function. If the largest subset $\X$ of $\Y$ such that $F\circ J_\X=J_\X$ exists, we call it the \emph{quantum set of fixpoints of} $F$, denoted by $\mathrm{Fix}(F)$. Similarly, if the largest subset $\X$ of $\Y$ such that  $F\circ J_\X\sqsubseteq J_\X$ exists, we call it the  \emph{quantum set of prefixpoints} of $F$, denoted by $\mathrm{Pre}(F)$. Finally, if the largest subset $\X$ of $\Y$ such that $F\circ J_\X\sqsupseteq J_\X$ exists, we call it  the \emph{quantum set of postfixpoints} of $F$, denoted by $\mathrm{Post}(F)$.
    \end{definition}
    Denoting $\Q\{X\}$ for the quantum subset of $\Y$ consisting of the atom $X\atomof\Y$, we can show that the largest subset $\mathrm{Post}(F)$ of postfixpoints of a monotone endofunction $F$ on a quantum poset $\Y$ exists and is determined by $\At(\X)=\{X\atomof\Y\colon J_{\Q\{X\}}\sqsubseteq F\circ J_{\Q\{X\}}\}$. Similarly, one can show that $\mathrm{Pre}(F)$ and $\mathrm{Fix}(F)$ exist. 
    
    Classically, the postfixpoints $P$ of a monotone map $f\colon Y\to Y$ on a suplattice form a suplattice. This can be seen by a follows. $P=\{y\in Y\colon y\leq f(y)\}$. Let $S\subseteq P$, and let $x=\sup S$. Then by monotonicity of $f$, we have $y\leq f(y)\leq f(x)$ for each $y\in S$, hence $x=\bigvee S\leq f(x)$, showing that $P$ is closed under suprema, so a suplattice. In a similar way, we can prove:
\begin{proposition}
  Given a monotone endofunction $F\colon\Y\to\Y$ on a quantum suplattice $\Y$, the quantum set  $\mathrm{Post}(F)$ of postfixpoints of $F$ is a quantum suplattice with respect to the relative order.
\end{proposition}

Applying Theorem \ref{thm:quantum suplattice if quantum inflattice}, it is easy to show that $\mathrm{Pre}(F)$ also is a quantum suplattice with respect to the relative order. Finally, we can show that $F$ restricts and corestricts to a monotone function $F|_\mathrm{Post(F)}^{\mathrm{Post}(F)}$ and that $\mathrm{Fix}(F)=\mathrm{Pre}\left (F|_\mathrm{Post(F)}^{\mathrm{Post}(F)}\right)$ from which we conclude:
\begin{theorem}[Quantum Knaster-Tarski Theorem]
Let $F:\Y\to\Y$ be a monotone endomap on a quantum suplattice $\Y$. Then the quantum set $\mathrm{Fix}(F)$ of fixpoints of $F$ is a quantum suplattice with respect to the relative order.
\end{theorem}

\section{The relation between ordinary suplattices and quantum suplattices}

Quantum posets are noncommutative generalizations of ordinary posets, since we have a fully faithful functor $`(-)\colon\POS\to\qPOS$. This functor restricts and corestricts to a functor $\CPO\to\qCPO$, hence also 
quantum cpos are genuine noncommutative generalizations of ordinary cpos. Remarkably, quantum suplattices are a not noncommutative generalizations of ordinary suplattices, but only noncommutative versions of suplattice: the functor $`(-)\colon\POS\to\qPOS$ does not restrict and corestrict to a functor $`(-)\colon\Sup\to\qSup$. Indeed, if $B$ denotes the four-element Boolean algebra, which is clearly an ordinary suplattice, then $`B$ is not a quantum suplattice. In order to see this, we first recall Section \ref{sec:enrichment}, in which it was stated that the limit of any collection of functions with the same domain into a quantum suplattice always exists. We will give a collection of functions into $`B$ that does not have a limit. Write $B=\{0,a,b,1\}$, where $a^\perp=b$, and denote the order on $B$ by $\sqsubseteq$, so $0\sqsubseteq a,b\sqsubseteq 1$. 
Let $(\X,\qo)=`(B,\sqsubseteq)$.  Let $\H$ denote the atomic quantum set whose single atom $H$ is two-dimensional. One can show that any function $K\colon\H\to\X$ is of the form $K(H,\CC_x)=L(H,\CC_x)r_x$ for each $x\in B$, where $r_0,r_a,r_b,r_1$ are mutually orthogonal projections on $H$ whose sum equals $1_H$. Then $(\qo\circ K)(H,\CC_x)=L(H,\CC_x)\sum_{y\sqsubseteq x}r_y$.
So $(\qo\circ K)(H,\CC_x)$ is of the form $L(H,\CC_x)s_x$ for some projection $s_x$ of $H$, and clearly $s_0$, $s_a$, $s_b$ and $s_1$ mutually commute, because $r_0$, $r_a$, $r_b$, $r_1$ are mutually orthogonal.

Let $p$ and $q$ be two distinct nontrivial noncommuting projections on $H$. For instance, in the standard basis of $H$, let $p=\begin{pmatrix}
1 & 0\\
0 & 0
\end{pmatrix}$ and $q=\frac{1}{2}\begin{pmatrix}
1 & 1 \\
1 & 1 
\end{pmatrix}$. Then $p$ and $q$ are also not orthogonal, whence, $pq\neq 0$. However, since $p$ and $q$ are both atomic projections, we have $p\wedge q=0$.

Let $F,G\colon\H\to\X$ be functions defined as follows. The nonzero components of $F$ are given by $F(H,\CC_0)=L(H,\CC_0)p$ and $F(H,\CC_a)=L(H,\CC_a)p^\perp$, whereas the nonzero components of $G$ are given by $G(H,\CC_0)=L(H,\CC_0)q$, $G(H,\CC_b)=L(H,\CC_b)q^\perp$. Then 
\[(\qo\circ F)(H,\CC_x)=\begin{cases}
L(H,\CC_x)p, & x=0,b;\\
L(H,\CC_x), & x=a,1,
\end{cases}\]
\[(\qo\circ G)(H,\CC_x)=\begin{cases} L(H,\CC_x)q, & x=0,a;\\
L(H,\CC_x), & x=b,1.
\end{cases}\]
hence 
\[ (\qo\circ F\wedge \qo\circ G)(H,\CC_x)=\begin{cases} 
0, & x=0\\
L(H,\CC_a)q, & x=a;\\
L(H,\CC_b)p, & x=b;\\
L(H,\CC_1), & x=1.
\end{cases}\]
Assume that there is a function $K\colon\H\to\X$ such that $\qo\circ R=(\qo\circ F)\wedge(\qo\circ G)$. Since $(\qo\circ K)(H,\CC_a)=L(H,\CC_a)q$ and $(\qo\circ K)(H,\CC_b)=L(H,\CC_b)p$, it follows that $p$ and $q$ should commute, which contradicts our assumptions. We conclude that there is a $\kK\subseteq\qSet(\H,\X)$ such that $\lim\kK$ does not exists, namely $\kK=\{F,G\}$. Hence, $\X$ cannot be a quantum suplattice, for which all limits should exist. 

The main reason why suplattices are not quantum suplattices lies in the fact that the $`(-)$ functors do not commute with $\P$ and $\D$. Let $S$ be a set. Then $`P(S)$ does not equal $\P(`S)$, but can only be identified with the subset of one-dimensional atoms of $\P(`S)$. That this subset is proper follows for instance from the proof of \cite[Proposition 9.3]{Kornell18}, which asserts that $(\mathbf 1\uplus \mathbf 1)*(\mathbf 1\uplus \mathbf 1)$, which can be identified with $\P(`2)$, has uncountably many atoms. The same is true for $`D(S)$ if $S$ is a poset: $`D(S)$ can only be identified with the subset of one-dimensional atoms of $\D(`S)$.  We plan to investigate whether the quantum power set monad $\P$ is the right Kan extension of $`(-)\circ P\colon\Set\to\qSet$ along $`(-)\colon\Set\to\qSet$, where $P$ denotes the ordinary power set monad. If this is true, we also expect that the quantum down-set monad $\qDwn$ is the right Kan extension of $`(-)\circ\Dwn\colon\POS\to\qPOS$ along $`(-)\colon\POS\to\qPOS$.

It follows from the $\Sup$-enrichment of $\qSup$ that the one-dimensional atoms of a quantum suplattice form an ordinary suplattice. We expect that any ordinary suplattice is the subposet of one-dimensional atoms of some quantum suplattice. 

By extension, if quantum suplattices indeed form the right notion that is needed to generalize topologies to the noncommutative setting, then it follows that an ordinary topology on an ordinary set is not a quantum topology, it might only be the classical part of a quantum topology. Moreover, it might be that there are several different quantum topologies on an ordinary set that have the same ordinary topology as classical part. If this is indeed the case, it is the question how to interpret this. Perhaps this would be a quantum feature, in the  same spirit as the result of two nonisomorphic graphs that are quantum isomorphic \cite{ATSERIASetal}.

\section{Future work}\label{sec:future}

The classical Knaster-Tarski Theorem implies the Cantor-Schr\"oder-Bernstein Theorem. We conjecture that a quantum version of Cantor-Schr\"oder-Bernstein can be derived from the quantum Knaster-Tarski Theorem.
\begin{conjecture}
Let $F\colon\X\to\Y$ and $G\colon\Y\to\X$ be injective functions between quantum sets. Then there exists a bijection between $\X$ and $\Y$. 
\end{conjecture}
The biggest potential obstacle is that an atom of the quantum power set of a quantum set $\X$ does not directly correspond to a subset of $\X$, as is the case in the classical case.

A quantum Cantor-Schr\"oder-Bernstein Theorem can be reformulated in terms of operator algebras: 
\begin{conjecture}
Given surjective normal unital $*$-homomorphisms $\varphi:M\to N$ and $\psi:N\to M$ between hereditarily atomic von Neumann algebras, there must exist a $*$-isomorphism between $M$ and $N$.
\end{conjecture}

Furthermore, based on the fact that any quantum suplattice is a quantum inflattice, we expect:
\begin{conjecture}$\qSup$ can be equipped with a monoidal product that makes it $*$-autonomous.
\end{conjecture}

\vspace{-.15in}
\section*{Acknowledgements}  
We thank Andre Kornell for his advice and for inspiring us to further develop the theory of quantum sets. Furthermore, we thank Isar Stubbe for helping us to understand the connection between quantum sets and quantaloids better. Furthermore, we thank the reviewers for their comments. This research is supported by grants VEGA 2/0142/20 and 1/0036/23 and by the Slovak Research and Development Agency under the contracts APVV-18-0052 and APVV-20-0069.


\bibliographystyle{eptcs}
\bibliography{qCPOrefs}

\end{document}